\documentclass[journal,10pt,draftclsnofoot,onecolumn]{IEEEtran}
\usepackage{url}
\usepackage{graphics}
\usepackage{amsmath,amsfonts,amssymb,epsfig}
\usepackage{lscape}
\usepackage{latexsym}

\newcommand{\MZ}{\mathbb{Z}}

\newcommand{\MQ}{\mathbb{Q}}

\newtheorem{theorem}{\bf Theorem}
\newtheorem{lemma}[theorem]{\bf Lemma}

\newtheorem{corollary}[theorem]{\bf Corollary}

\newtheorem{definition}[theorem]{\sc Definition}

\newenvironment{remark}{{\noindent \bf Remark:}}{}
\newenvironment{proof}{{\sc Proof.}}{\hspace*{\fill}$\Box$\par\vspace{4mm}}

\begin{document}

\title{On a Connection between Ideal Two-level Autocorrelation and Almost Balancedness of $p$-ary Sequences*}

\author{Yuri L. Borissov\\
\thanks{Yuri L. Borissov is with the Dept. of Mathematical Foundations of Informatics, Institute of Mathematics and Informatics,
Bulgarian Academy of Sciences, 8 G. Bontchev Str., 1113 Sofia, Bulgaria.

* This paper was submitted to Comptes rendus de l'Academie Bulgare des Sciences.}
}

\maketitle

\begin{abstract}\noindent
In this correspondence, for every periodic $p-$ary sequence satisfying ideal two-level autocorrelation property the existence of an element of the field ${\bf GF}(p)$ which appears one time less than all the rest that are equally distributed in a period of that sequence, is proved by algebraic method. In addition, it is shown that such a special element might not be only the zero element but as well arbitrary element of that field.

\end{abstract}


\section{\bf Introduction}
\label{sect1} Balancedness and low out-of-phase autocorrelation
values are properties which one expects to find in randomly
generated sequences. However, in some engineering applications
(e.g., CDMA communication, radar, cryptography etc.) the
so-called pseudo-random sequences (which are generated
deterministically) are preferable because of the possible
deviations and to avoid data storage requirements. Among these
sequences we distinguish the {\bf m}-sequences, the GMW sequences
and their derivatives (both binary and non-binary) and many other
types of sequences having the properties which mimic those of
truly random sequences (see, e.g., \cite{Golomb} --
\cite{Tilborg}). In this correspondence, the relationship between two of
these properties: the ideal two-level autocorrelation property and
the almost balancedness in $p-$ ary case, is investigated. In the
binary case, it has been proved that the number of ones and zeroes
in a sequence with the two-level autocorrelation differ by 1 (see,
e.g., \cite{Golomb} or \cite{Tilborg}[Chapter $3$]). A similar
situation is observed in the general case ($p$ arbitrary prime)
and it is treated algebraically in this article.

The paper is organized as follows. In next section, we recall some
definitions and preliminaries, and in Section $3$, we present
our results.

\section{\bf Definitions and Preliminaries}
\label{sect2} Let ${\bf a} = \{a_{n}\}, n = 0,1,\ldots,N-1$ be a
sequence of length $N$ with entries from some finite set ${\bf S}$. For
arbitrary $c \in {\bf S}$, define $\mu_{c} = \vert n :\; a_{n} =
c,\; 0 \leq n \leq N-1 \vert$,  and call $\mu_{c}$ multiplicity of
$c$ (or frequency of appearance) in ${\bf a}$. Clearly, we have:
\begin{eqnarray}\label{eq.1}
\sum_{c \in {\bf S}}\mu_{c} = N
\end{eqnarray}

Sequences with entries from the finite field ${\bf GF}(p)$
for a prime $p$, are called sometimes $p-$ary sequences.
In this correspondence, we consider mostly pure periodic $p-$ary sequences.

In order to present our result we need to recall the following two
definitions.

\begin{definition}\label{balanced}
Let $p$ be a prime and $t$ be some positive integer. A $p-$ary
sequence ${\bf a}$ of period $N = pt-1$, is called almost balanced
if there exists $c \in {\bf GF}(p)$ with multiplicity $t-1$, while
all other elements of the field are of multiplicity $t$ in one
period of ${\bf a}$. We will also say that $c$ has an exceptional
frequency of appearance in a period of ${\bf a}$.
\end{definition}

Let $\omega$ denotes the $p-$th primitive root of unity $e^{2 \pi i/p}$. Then
${\cal R}_{p} = \{ 1,\omega,\ldots,\omega^{p-1}\}$ is the set of
all $p-$th roots of unity. By identifying ${\bf GF}(p)$ with
$\MZ_{p}$, one can put into correspondence to any $p-$ary sequence ${\bf a} = \{a_{n}\}, n =
0,1,\ldots$ the complex sequence ${\bf s} = \{s(n) = \omega^{a_{n}}\}, n = 0,1,\ldots$
whose entries are from ${\cal R}_{p}$.

Further on, we shall use the notation $\alpha^{*}$ for the
conjugate of a complex number $\alpha$.
\begin{definition}(see, e.g., \cite{NoGo})\label{def.2}
The complex sequence  ${\bf s} = \{s(n)\}$ of period $N$ is said
to have the ideal two-level autocorrelation property if its
autocorrelation function $R_{\bf s}(k)$ satisfies the following:
\begin{eqnarray*}
{R_{\bf s}(k)} =  {{\hspace*{-0.00001cm} N, \;\;\;\;\; {\rm if}\;\; k \equiv 0\; mod\;N,} \atopwithdelims \{.  {\hspace*{-1cm} -1, \;\;\; {\rm  otherwise,}}}
\end{eqnarray*}
where $R_{\bf s}(k)$ is defined as follows:
\begin{equation*}
R_{\bf s}(k) = \sum_{n=0}^{N-1}s(n)s^{*}(n+k).
\end{equation*}
\end{definition}

When ${\bf a}$ is a $p-$ary sequence with corresponding complex
sequence possessing the ideal two-level autocorrelation property,
we shall say that ${\bf a}$ also possesses this property.

\section{\bf Results}
\label{sect3}

We start with the following
\begin{lemma}\label{lemma1}
Let ${\bf a} = \{a_{n}\}, n = 0,1,\ldots$ be a $p-$ary sequence
${\bf a}$ of period $N$ having the ideal two-level autocorrelation
property. Then there exists a positive integer $t$ such that $N = pt-1$.
\end{lemma}

\begin{proof}
Let ${\bf s} = \{s(n) = \omega^{a_{n}}\}, n = 0,1,\ldots$ be the
complex sequence corresponding to ${\bf a}$. Define
${\bf u} = \{u(n) = s(n)s^{*}(n+1)\}, n = 0,1,...,N-1$.
Obviously, $u(n) \in {\cal R}_{p}$, and therefore $u(n) =
\omega^{k_{n}}$ for some $0 \leq k_{n} < p$. Then according to the
definition of the ideal two-level autocorrelation property, we
have: $R_{\bf s}(1) = \sum_{n=0}^{N-1}u(n) = -1$, or equivalently
$\sum_{n=0}^{N-1}\omega^{k_{n}} + 1 = 0$. Let $\mu_{k}$ be the
multiplicity of $\omega^{k}$ in ${\bf u}$, for $0 \leq k < p$ .
Hence, from the previous equation, we get:
\begin{eqnarray*}
\sum_{k=p-1}^{1}\mu_{k}\omega^{k} + \mu_{0} + 1 = 0
\end{eqnarray*}
Thus, $\omega$ is a root of the polynomial $f({\bf x}) =
\sum_{k=p-1}^{1}\mu_{k}{\bf x}^{k} + \mu_{0} + 1$. Since the minimal polynomial of
$\omega$ (over the field of rational
numbers $\MQ$) is the $p-$th cyclotomic polynomial $\Phi_{p}({\bf x}) = {\bf x}^{p-1} + {\bf x}^{p-2} +
\ldots +{\bf x} + 1$, applying the well-known fact from Abstract Algebra (see, e.g., \cite{Obr}),
we conclude that $\Phi_{p}({\bf x})$ divides $f({\bf x})$.
Further, since $p - 1 = deg(\Phi_{p}) \geq
deg(f)$ and $f$ has integer coefficients, there exists some
positive integer $t$ for which $\mu_{k} = t$ when $1 \leq k < p$
and $\mu_{0} + 1 = t$. Finally, by equation (\ref{eq.1}), we have $N =
\sum_{k=0}^{p-1}\mu_{k} = t-1 +(p-1)t = pt-1$, and the proof is
completed.
\end{proof}

\begin{remark} The above Lemma shows that a $p-$ary
sequence possessing the ideal two-level autocorrelation property
cannot be strictly balanced.
It is well known that if a binary sequence has
the ideal two-level autocorrelation, it must have a period $N$
with $N \equiv -1\;(mod\;4)$ (see, e.g., \cite{KimSong}). Of
course, this fact is stronger than the claim of Lemma
\ref{lemma1} in case $p=2$.
Also, in case $p > 2$, for all known periodic sequences with the
ideal two-level autocorrelation, the period $N$ is of the form
$p^{m}-1$ for some positive $m$, while in the binary case there
exist examples when it is of different type (see, e.g., \cite{Go}).
But here, we will not discuss the details of this topic.
\end{remark}

\begin{lemma}\label{lemma2}
Let $A_{n}, 0 \leq n < k$ be $k$ positive integers which obey the two equations:
\begin{eqnarray}\label{eq.3}
\sum_{n=0}^{k-1}(A_{n} - A_{n+1})^{2} = 2,\;if \; k \geq 2
\end{eqnarray}
\begin{eqnarray}\label{eq.4}
\sum_{n=0}^{k-1}(A_{n} - A_{n+2})^{2} = 2,\;if \;k \geq 4
\end{eqnarray}
where $n+j,\; j = 1,2$ are taken modulo $k$. Then there exists an
index $m:\;0 \leq m < k$ and a positive integer $A$, such that
$A_{m} = A \pm 1$, while all other $A_{n}$s are equal to $A$.
\end{lemma}
\begin{proof}
Imagine that $A_{0},A_{1}.\ldots,A_{k-1}$ are written alongside a
circle. Since the differences $A_{n}-A_{n+j},\;j=1,2$ are integer
numbers, in each of the given equations there exist exactly two
terms $(A_{n}-A_{n+j})^{2}$ equal to $1$, while all other terms
have to be equal to $0$.  Thus equation (\ref{eq.3}) implies that the
$A_{n}$s are divided into two nonempty groups of equal adjacent
numbers and the common values (for each group its own) differ by
$1$ (So, the proof is completed in cases $k=2,3$). Moreover,
taking into account equation (\ref{eq.4}), we conclude that one of the
groups has to contain a single number, since otherwise four terms
$(A_{n}-A_{n+2})^{2}$ will be equal to $1$. This completes the
proof in case $k > 3$.
\end{proof}

Furthermore, we prove the following theorem.

\begin{theorem}\label{th.1}
Let ${\bf s} = \{s(n)\}, n = 0,1,\ldots$ be complex sequence of
period $N$ having the ideal two-level autocorrelation property.
Then it holds $\vert \sum_{n=0}^{N-1} s(n) \vert = 1$.
\end{theorem}
\begin{proof}
By straightforward computations taking into account the definition
of autocorrelation function, we have:
\begin{eqnarray*}
\sum_{k=0}^{N-1}R_{\bf s}(k) = \sum_{k=0}^{N-1}\sum_{n=0}^{N-1}s(n)s^{*}(n+k) =
\sum_{n=0}^{N-1}s(n)(\sum_{k=0}^{N-1}s^{*}(n+k)) =
\end{eqnarray*}
\begin{eqnarray*}
\sum_{n=0}^{N-1}s(n)(\sum_{k=0}^{N-1}s^{*}(k)) =
\sum_{n=0}^{N-1}s(n) \times \sum_{k=0}^{N-1}s^{*}(k) = \sum_{n=0}^{N-1}s(n) \times (\sum_{k=0}^{N-1}s(k))^{*} =
\vert \sum_{n=0}^{N-1} s(n) \vert^{2}.
\end{eqnarray*}
On the other hand, by Definition \ref{def.2} we have that:
$\sum_{k=0}^{N-1}R_{\bf s}(k) = 1$. So, $\vert \sum_{n=0}^{N-1} s(n) \vert^{2} = 1$ which completes the proof.
\end{proof}
From Theorem \ref{th.1}, one can immediately derive the
following.
\begin{corollary}\label{cor.1}
Let ${\bf a}$ be a $p-$ary sequence of period $N$, and $\mu_n$ be
the frequency of appearance of $n,\; 0 \leq n \leq p-1,$ in one
period of ${\bf a}$. If in addition, ${\bf a}$ has the ideal two-level autocorrelation property then it holds:
\begin{eqnarray}\label{eq.2}
\vert \sum_{n=0}^{p-1} \mu_{n}\omega^{n} \vert^{2} = 1.
\end{eqnarray}
\end{corollary}

Now, we will prove the main theorem of this correspondence.
\begin{theorem} \label{th.2}
Let the complex sequence ${\bf s} = \{s(n)\}, n = 0,1,\ldots$
corresponding to a $p-$ary sequence ${\bf a}$ of period $N$ has
the ideal two-level autocorrelation property. Then ${\bf a}$ is
almost balanced.
\end{theorem}
\begin{proof}
Let $\mu_{n}$ be the multiplicity of $n \in \MZ_{p}$ in one period
of ${\bf a}$. We will rewrite the left-hand side of the equation
(\ref{eq.2}) from Corollary \ref{cor.1} as follows:
\begin{eqnarray*}
\vert \sum_{n=0}^{p-1} \mu_{n}\omega^{n} \vert^{2} = (\sum_{n=0}^{p-1} \mu_{n}\omega^{n})(\sum_{n=0}^{p-1} \mu_{n}\omega^{n})^{*} = (\sum_{n=0}^{p-1} \mu_{n}\omega^{n})(\sum_{n=0}^{p-1} \mu_{n}\omega^{p-n}) = \pi(\omega)
\end{eqnarray*}
Further, arranging $\pi(\omega)$ according to the powers of $\omega$, we get:
\begin{eqnarray*}
\pi(\omega) = \sum_{n=0}^{p-1}B_{n}\omega^{p-n},
\end{eqnarray*}
where
\begin{eqnarray}\label{eq.5}
B_{n} = \sum_{k=0}^{p-1}\mu_{k}\mu_{k+n},\; n = 0,1,\ldots,p-1.
\end{eqnarray}
On the other hand, by equation (\ref{eq.2}), we have: $\pi(\omega) = 1,$ or
\begin{eqnarray*}
\sum_{n=1}^{p-1}B_{n}\omega^{p-n}+ \sum_{n=0}^{p-1}\mu_{n}^{2} - 1 = 0.
\end{eqnarray*}
So, $\omega$ is a root of polynomial $f({\bf x}) =
\sum_{n=1}^{p-1}B_{n}{\bf x}^{p-n}+ \sum_{n=0}^{p-1}\mu_{n}^{2} -
1$. Proceeding with polynomials $f$ and $\Phi_{p}$ like in the proof of
Lemma \ref{lemma1}, we conclude that there exists a positive
integer $C$ which satisfies the following equations:
\begin{eqnarray*}
B_{n} = C,\;for\; n = p-1,\ldots,2,1
\end{eqnarray*}
\begin{eqnarray*}
\sum_{n=0}^{p-1}\mu_{n}^{2} = C+1
\end{eqnarray*}
Subtracting from the last equation the two formers and taking into account equation (\ref{eq.5}), we get:
\begin{eqnarray*}
\sum_{n=0}^{p-1}\mu_{n}^{2} - B_{j} = {1 \over 2} \sum_{n=0}^{p-1}(\mu_{n} - \mu_{n+j})^{2} = 1,\;for\;j=1,2,
\end{eqnarray*}
or equivalently
\begin{eqnarray}
\sum_{n=0}^{p-1}(\mu_{n} - \mu_{n+j})^{2} = 2,\;for\;j=1,2.
\end{eqnarray}
Now, by Lemma \ref{lemma2}, it follows the existence of an index
$m:\;0 \leq m \leq p-1$ and a positive integer $\mu$ such that
$\mu_{m} = \mu \pm 1$, while all other $\mu_{n}$s are equal to
$\mu$. Finally, Lemma \ref{lemma1} excludes the possibility
that $\mu_{m} = \mu +1$ when $p > 2$ (Note that if $p=2$ both
possibilities look the same). Therefore, the sequence ${\bf a}$ is
almost balanced.
\end{proof}

To the best of our knowledge, all known constructions of $p-$ary
sequences, $p > 2$, possessing the ideal two-level autocorrelation
property have the peculiarity to provide almost balanced sequences
for which the frequency of appearance of zero is exceptional, i.e.
zero appears once less than the other elements of ${\bf GF}(p)$
(see, e.g., \cite{SchWe} -- \cite{FanDar}). Moreover, in another
two papers \cite{GongSong1} and \cite{GongSong2} (among other things)
it is claimed that "the balance property has been proved in \cite{LudGong}
assuming only the ideal two-level autocorrelation function (whenever $q = p > 2$ is an odd prime)".
And the balance property of a $p-$ary sequence of period $p^{m}-1$ is defined as
that zero appears $p^{m-1}-1$ times while any nonzero element of ${\bf GF}(p)$
appears $p^{m-1}$ times in one period.

But as the next theorem shows that feature is not common
for sequences with the ideal two-level autocorrelation property whatever might be the period.

\begin{theorem}\label{th.3}
Let ${\bf a^{\prime}} = \{a^{\prime}_{n}\}, n = 0,1,\ldots$ be a $p-$ary sequence
of period $N$ having the ideal two-level autocorrelation property
and let the element $c^{\prime} \in {\bf GF}(p)$ be with exceptional frequency of
appearance in one its period.
Define ${\bf a}$ as $\{a^{\prime}_{n} - c^{\prime} + c\}, n = 0,1,\ldots$,
where $c$ is an arbitrary element of ${\bf GF}(p)$.
Then the latter sequence satisfies the ideal two-level autocorrelation property too,
having the element $c$ with exceptional frequency of appearance in one period.
\end{theorem}
\begin{proof}
Since by Theorem \ref{th.2}, ${\bf a}^{\prime}$ is an almost balanced
sequence then ${\bf a}= \{a^{\prime}_{n} - c^{\prime} + c\}, n = 0,1,\ldots$ is an
almost balanced too, but of course, instead of $c^{\prime}$ the frequency
of appearance of $c = c^{\prime} - c^{\prime} + c$ is the
exceptional one. Let ${\bf s}$ and ${\bf s}^{\prime}$ be the
complex sequences corresponding to ${\bf a}$ and ${\bf
a}^{\prime}$, respectively. Denote by $w = \omega^{c - c^{\prime}}$
Then by the definition of autocorrelation function (see,
Definition \ref{def.2}) for any $k$, we have the following:
\begin{equation*}
R_{{\bf s}}(k) = \sum_{n=0}^{N-1}s(n){s}^{*}(n+k) =  \sum_{n=0}^{N-1}(s(n)^{\prime}\; w) \times (s^{\prime\;*}(n+k)\; w^{*}) = R_{\bf s^{\prime}}(k),
\end{equation*}
since the product $ w \times  w^{*}$ equals to $1$. Thus, the
autocorrelation function of ${\bf s}$ coincides with that one
of ${\bf s}^{\prime}$ and therefore ${\bf a}$ possesses the
ideal two-level autocorrelation property as well.
\end{proof}

In other words, the above theorem states that together with any $p-$ary
sequence satisfying the ideal two-level autocorrelation property there exists
a whole one-parametric family of cardinality $p$ containing sequences of this kind
(and which are not cyclic replicas of the primary sequence).
This fact might be useful to vary, for instance, CDMA communication.

\section*{\bf Acknowledgments}
The author would like to thank Stefan M. Dodunekov, Ivan N. Landjev and Svetla Nikova for helpful discussions and comments which substantially improve the presentation of the results.

\bibliography{}

\end{document}